%% file: icassp_draft_short.tex
\documentclass{article}
\usepackage{spconf,amsmath,graphicx,hyperref}

\usepackage[noadjust]{cite}
\usepackage{amsmath}
\usepackage{amssymb}
\usepackage{amsthm}
\usepackage{mathtools}
\usepackage{bbm}
\usepackage{algorithm}
\usepackage[noend]{algpseudocode}
\usepackage{bm}
\usepackage{comment}
\usepackage{tabularx,booktabs}
\usepackage[capitalize]{cleveref}
\usepackage[caption=false]{subfig}

\usepackage{subcaption}
\usepackage{caption}
\usepackage{pgfplots}
\pgfplotsset{compat=1.15}
\usepgfplotslibrary{fillbetween}
\usetikzlibrary{patterns,arrows,plotmarks}
\usepgfplotslibrary{groupplots}
\pgfdeclarelayer{background}
\pgfsetlayers{background,main}
\usetikzlibrary{automata,positioning}
\usetikzlibrary{decorations}
\usetikzlibrary{shapes.arrows}
\usetikzlibrary{tikzmark}
\usetikzlibrary{calc}
\usetikzlibrary{decorations.markings}
\algrenewcommand\algorithmicindent{10pt}
\usepgfplotslibrary{colorbrewer}

\usepackage{tikz}
\usepackage{graphicx}

\usepackage[utf8]{inputenc}

\newtheorem{proposition}{Proposition}

\newcommand{\mc}[1]{\mathcal{#1}}   
\DeclareMathOperator*{\argmax}{arg\,max}    
\DeclareMathOperator*{\argmin}{arg\,min}    

\newcommand{\1}{\mathbbm{1}}

\usepackage[acronym]{glossaries}

\input{acronyms}

\definecolor{okabeBlue}{RGB}{0,114,178}
\definecolor{okabeOrange}{RGB}{230,159,0}
\definecolor{okabeGreen}{RGB}{0,158,115}
\definecolor{okabeVermillion}{RGB}{213,94,0}

\begin{document}
\title{Risk-Aware Online Conformal State Probing}


\name{Pietro Talli$^{\star}$ \qquad Petar Popovski$^{\dagger}$ \qquad Osvaldo Simeone$^{\star,\dagger}$\thanks{The work of O. Simeone was supported by an Open Fellowship of the EPSRC (EP/W024101/1) and by the EPSRC (EP/X011852/1). The work of P. Talli and O. Simeone was supported by the ERC (No. 101198347).}}
  \vspace{-4 cm}
  \address{$^{\star}$Institute for Intelligent Networked Systems, Northeastern University London \\
      $^{\dagger}$Connectivity Section, Department of Electronic Systems, Aalborg University \\
      Email: $\{$p.talli, o.simeone$\}$@northeastern.edu, petarp@es.aau.dk}

\maketitle

\begin{abstract}
AI-based autonomous agents, typically hosted at data centers, must acquire state information from robots or edge devices in order to issue informed control decisions. Managing uncertainty about the state is particularly consequential in safety-critical settings, in which average-case guarantees are insufficient. In this context, we study a sequential decision maker process that jointly decides which actions to take and when to probe given access to an arbitrary state prediction model. We propose \gls{ocsp}, an action and probing policy that certifies worst-case reliability levels  without relying on distributional assumptions. \gls{ocsp} is designed to provably control the \gls{mqe}, i.e., the fraction of instances where probing would have been beneficial, while minimizing the probing rate. \gls{ocsp} can be applied to existing pre-trained value-based control policies without requiring retraining or fine-tuning. We validate \gls{ocsp} through numerical simulations to verify theoretical guarantees  and to assess performance trade-offs as a function of the calibration of the state predictor.
\end{abstract}
 

\glsresetall

\section{Introduction}

Recent advances in AI, including \glspl{llm} and \glspl{vla} \cite{black2026real}, have equipped autonomous agents, typically hosted at data centers, with greater capabilities and generalization skills. However, the deployment of AI-based autonomous agents increasingly demands strict reliability guarantees on performance \cite{li2023behavior, rabanser2026towards,rabanser2026towards}. A central difficulty in ensuring reliable operation is that agents residing on the cloud must act on system states that are not promptly available, as telemetry from robots or edge devices can be too costly to query continuously and in a timely fashion (see Fig.~\ref{fig:model}). As a result, the agent must act under partial observability, actively deciding when to probe the state in order to reduce its uncertainty about the environment before committing to an action \cite{gunduz2023timely}. 

\begin{figure}
    \centering
\begin{tikzpicture}
    \node at (0,0) {\includegraphics[width=0.5cm]{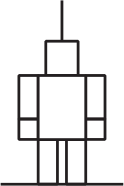}};
    \draw[dashed] (1,-0.5) ellipse (0.45cm and 0.35cm);

    \draw[dashed] (2,-1) ellipse (0.6cm and 0.5cm);

    \draw[dashed] (3,-1.5) ellipse (0.7cm and 0.6cm);

    \node at (3,0.6) {\includegraphics[width=1cm]{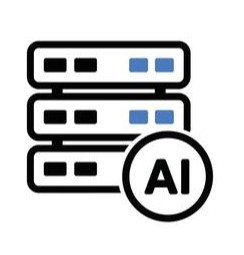}};

    \node at (3,-1.5) {\includegraphics[width=0.5cm]{files/robot.png}};


    \draw[dashed] (3.9,-1) ellipse (0.45cm and 0.35cm);
    \draw[dashed] (5,-1.2) ellipse (0.6cm and 0.5cm);
    \draw[dashed] (6,-1.5) ellipse (0.45cm and 0.35cm);
    
    \node at (5,-1.2) {\includegraphics[width=0.5cm]{files/robot.png}};

    \draw[dashed] (5,-1.2) ellipse (0.6cm and 0.5cm);

    \draw[dotted] (0,0) -- (1.2,-0.4) -- (2.1,-1.2) -- (3,-1.5) -- (3.7,-1.1) -- (5,-1.2);

    \fill[black] (0,0) circle (0.05cm);
    \fill[black] (3,-1.5) circle (0.05cm);
    \fill[black] (5,-1.2) circle (0.05cm);


    \draw[->, thick, red] (0.2,0.1) -- (2.5,0.6); 
    \draw[->] (2.5,0.5) -- (0.2,0); 

    \draw[->, thick, red] (2.95,-1.2) -- (2.95,0.2); 
    \draw[->] (3.05,0.2) -- (3.05,-1.2); 

    \draw[->, thick, red] (4.9,-1) -- (3.4,0.45); 
    \draw[->] (3.4,0.6) -- (5,-0.9); 

    \draw[->, thick, red] (4.5,0.8) -- (5,0.8);
    \draw[->] (4.5,0.48) -- (5,0.48);
    \node at (6,0.8) {\footnotesize state probing};
    \node at (6.0,0.5) {\footnotesize control actions};
    \draw[dashed] (4.75,0) ellipse (0.3cm and 0.25cm);
    \node[text width=2.4cm] at (6.4,0) {\footnotesize uncertainty for \\[-4pt] current belief};
    
\end{tikzpicture}
    \caption{Example of an agent that navigates the environment according to an estimated belief on the current state, actively probing the state to control the outage rate. \vspace{-0.5 cm}}
    \label{fig:model}
\end{figure}

 Management of uncertainty is particularly consequential for reliability-critical applications, in which standard average-case guarantees are insufficient. In fact, a policy that performs well in expectation, but occasionally fails catastrophically, may be unacceptable. Ideally, reliability constraints must hold even in the worst case, without distributional assumptions on the environment or observation noise. 

In this work, we study a general sequential decision-making problem in which an intelligent agent jointly decides which control actions to execute and when to probe the state  (Sec. \ref{sec:2}). Fig.~\ref{fig:model} illustrates the example of a robot controlled by a cloud based agent. At each time step, a state-estimation model at the agent produces an uncertainty estimate for the current state, which is used by the agent to determine control and probing actions. Reliability requirements are imposed by constraining the \emph{\gls{mqe}} rate, i.e., the fraction of steps in which probing would have been beneficial but the state was not probed \cite{kiyani2026strategic}.

To ensure provable control of the \gls{mqe}, we introduce \emph{\gls{ocsp}}, an action and probing policy that wraps around any pre-trained state predictor (Sec.~\ref{sec:3}). \gls{ocsp} builds on online risk control, which provides deterministic, distribution-free coverage and risk guarantees for sequentially observed data \cite{angelopoulos2023conformal, gibbs2024conformal, lekeufack2024conformal, zecchin2026prediction}. Unlike this line of work, however, \gls{ocsp} does not aim to calibrate a predictor \cite{shafer2008tutorial} or to defer from cloud to edge \cite{kiyani2026strategic}, but to certify and control the reliability of a closed-loop decision-making mechanism that jointly governs \textit{which actions} to take and \textit{when} to probe. As a use case, \gls{ocsp} is applied to an existing value-based control policy in a \gls{mdp},  providing worst-case performance guarantees for a pre-trained agent (Sec. \ref{sec:experiments}).

\section{System Model}\label{sec:2}

We consider a stateful sequential decision process whose state at discrete time $t$ is given by $s_t\in\mc{S}$, where $\mc{S}$ is a discrete finite set. As in the example of Fig.~\ref{fig:model}, at each time step $t$, the agent does not have direct access to state $s_t$, but it can acquire state $s_t$ through a probing action $P_t\in\{0,1\}$, where $P_t=1$ indicates that the state is probed. The sequence of past probed states is used to maintain an estimate of the current state in the form of a probability $b_t\in\Delta_{\mc{S}}$ given by
\begin{equation}
\label{eq:belief_function}
b_t = f_t(\{(t',s_{t'}): P_{t'}=1\}_{t'<t}),
\end{equation}
where $\Delta_{\mc{S}}$ is the simplex of probability distributions over the set $\mc{S}$ and $f_t$ is a sequence of functions. Functions $f_t$ may be implemented as a pre-trained predictor, and are not subject to optimization. As discussed in Sec.~\ref{sec:experiments}, the belief may be evaluated using a probabilistic transition model for the state.

At every time step $t$, the agent interacts with the environment, taking an action $a_t\in\mc{A}$ as a function of the current belief $b_t$, where $\mc{A}$ is a discrete finite set. This yields the control action $ a_t = \pi_t^{\mc{A}}(b_t),$ with \textit{action policy} $\pi^{\mc{A}}_t$ to be optimized. In general, the next state $s_{t+1}$ can be affected by the previous actions $a_{t'}$ with $t'\leq t$, and we do not make any assumption on the functional or probabilistic form of this dependence.

Upon taking action $a_t$, the agent collects the utility $U_t(s_t,a_t)\in[U_{\mathrm{min}}, U_\mathrm{max}]$, where $U_t(\cdot,\cdot)$ is an arbitrary known function, and we have bounds $U_{\mathrm{min}}\geq 0$ and $U_{\mathrm{max}}<\infty$. Thus, given the correct state $s_t$, the optimal action is 
\begin{equation}
\label{eq:optimal_action}
    a_t^* \in \argmax_{a\in\mc{A}} U_t(s_t,a), 
\end{equation}
which corresponds to the optimal utility value $U_t^*=U_t(s_t,a^*)$. As further discussed in Sec.~\ref{sec:experiments}, the utility $U_t(s,a)$ may be obtained as the state-action value function of a pre-trained value-based agent.

For all times $t$ with $P_t=1$, the state $s_t$ is known and the action can be optimally set to $a_t=a^*_t$. In contrast, at times with no probing, the action $a_t$ taken by the agent generally incurs a \textit{regret} compared to the optimal action $a^*_t$ in \eqref{eq:optimal_action}, which is given by the difference 
\begin{equation}
    R_t(s_t,\hat{a}_t) = U^*_t - U_t(s_t,\hat{a}_t). 
\end{equation}

To quantify the acceptable performance loss caused by incomplete state probing, we introduce a tolerance threshold $R_{\mathrm{max}} \leq U_\mathrm{max}-U_{\mathrm{min}}$, representing the maximum admissible regret. Violations of this threshold are captured by the binary \emph{outage} variable
\begin{equation}
\label{eq:outage_flag}
    O_t = \mathbbm{1}\!\left[R_t(s_t,\hat{a}_t) > R_{\mathrm{max}}\right].
\end{equation}
Importantly, since the agent does not observe the state unless it actively probes it, the outage signal $O_t$ is itself observed only when a probe is issued, i.e., when $P_t=1$.

We model the decision of whether to probe as a stochastic \emph{probing policy} $\pi_t^{\mc{P}}(b_t) \in [0,1]$, with $P_t \sim \mathrm{Bernoulli}(\pi_t^{\mc{P}}(b_t))$ denoting the decision to probe at time $t$. Overall, given fixed belief functions in \eqref{eq:belief_function}, we aim at designing the action policy $\pi_t^{\mc{A}}$ and the probing policy $\pi_t^{\mc{P}}$. The design goal is to control the \gls{mqe} \cite{kiyani2026strategic}, i.e., the fraction of instances in which probing would have been beneficial, since $O_t=1$, but was not performed, i.e., $P_t=0$. Formally, the \gls{mqe} is defined as
\begin{equation}
\label{eq:mqe}
    \mathrm{MQE}_T = \frac{1}{N_T} \sum_{t=1}^{T} O_t\,(1 - P_t),
\end{equation}
where $N_T = \sum_{t=1}^{T} O_t$ is the number of outages in the sequence, with the convention $\mathrm{MQE}_T = 0$ when $N_T = 0$. A zero \gls{mqe} can be obtained by a trivial policy that always probes the state, setting $P_t=1$ for all times $t$. Therefore, our design goal is to control the \gls{mqe}, while making a best effort at minimizing the probing rate $\sum_{t=1}^T P_t/T$. 

Formally, we aim to ensure that, for any sequence of states $\{s_t\}_{t=1}^T$, given any fixed belief functions $\{f_t\}_{t=1}^T$, the \gls{mqe} does not exceed a target level $\alpha\in[0,1]$ when $T$ is sufficiently large. This requirement is expressed by the inequality  
\begin{equation}
\label{eq:time_b}
    \mathrm{MQE}_T \leq \alpha + o_{N_T}(1),
\end{equation}
where the additive term $o_{N_T}(1)$ vanishes as the number of outages $N_T$ grows.

\section{Online Conformal State Probing}\label{sec:3}
In this section we present \gls{ocsp}, an online action-probing policy that provides worst-case guarantees on the \gls{mqe} as per the requirement~\eqref{eq:time_b}, while attempting to minimize the probing rate. 

\noindent \textbf{OCSP.} At each time step $t$, \gls{ocsp} first evaluates a candidate action $\hat{a}_t$ by minimizing the regret under the current belief $b_t$, i.e.,
\begin{equation}
\label{eq:candidate_a}
    \hat{a}_t = \argmin_{a\in\mc{A}} \sum_{s\in\mc{S}} b_t(s) R_t(s,a).
\end{equation}
We emphasize that, since the belief is an arbitrary estimate of the state, the action in \eqref{eq:candidate_a} does not provide any optimality guarantee. The action policy $\pi_t^{\mc{A}}$ then sets $a_t=\hat{a}_t$ when $P_t=0$, while setting $a_t=a^*_t$ when $P_t=1$.

\begin{algorithm}[t]
\caption{Online Conformal State Probing (\textsc{OCSP})}
\label{alg:ocsp}
\begin{algorithmic}[1]
\footnotesize
\Require target $\alpha\in[0,1]$; tolerance $R_{\max}\ge0$; step sizes $\eta>0$; exploration
  $\rho\in(0,1-\alpha]$; utility functions $\{U_t\}_{t=1}^T$; estimation functions $\{f_t\}_{t=1}^T$;
  initial $\mu_1$
\For{$t=1,2,\dots$}
  \State $b_t \gets f_t(\{s_{t'}:P_{t'}=1\}_{t'<t})$ \Comment{$\delta_{s_1}$ if t=1}
  \State $\hat{a}_t\gets\argmin_{a\in\mc{A}}\sum_{s\in\mc{S}}b_t(s)R_t(s,a)$
  \State $\hat{O}_t\gets \sum_{s\in\mc{S}} b_t(s)\1[R_t(s,\hat{a}_t) > R_{\mathrm{max}}]$
         \Comment{outage score}
  \State $p_t\gets \rho+(1-\rho)\,\1[\hat{O}_t\geq\mu_t]$
  \State draw $P_t\sim\mathrm{Bernoulli}(p_t)$
  \If{$P_t=0$}
     \State $a_t\gets \hat{a}_t$; \quad 
     $\mu_{t+1}\gets\mu_t$
  \Else \Comment{state probing (rule-driven or exploratory)}
     \State observe $s_t$;\quad
            compute $R_t$ and $O_t$
            \State $a_t \gets a^*_t;$
     \State update $\mu_t$ according to \eqref{eq:update_rule}
  \EndIf
\EndFor
\end{algorithmic}
\end{algorithm}

\gls{ocsp} decides whether to probe the state, setting $P_t=1$, or not, $P_t=0$, depending on a local estimate of the probability of an outage $O_t$ in \eqref{eq:outage_flag} for the candidate action $\hat{a}_t$. The outage probability is estimated as the probability of the outage event $O_t = \1 [R_t(s,\hat{a}_t) > R_{\mathrm{max}}]$ under the belief $b_t$, i.e.,  
\begin{equation}
    \hat{O}_t= \sum_{s\in\mc{S}} b_t(s)\1[R_t(s,\hat{a}_t) > R_{\mathrm{max}}]. 
\end{equation}
Specifically, \gls{ocsp} follows a probabilistic query policy $\pi_t^{\mc{P}}$ whereby probing is carried out with probability 
\begin{equation} 
\label{eq:p_t}
p_t = \pi_t^{\mc{P}}(b_t) =\rho+(1-\rho)\1[\hat{O}_t \geq \mu_t],
\end{equation}
i.e., $P_t \sim \mathrm{Bernoulli}(p_t)$,  
where $\rho\in(0,1-\alpha]$ is an exploration probability and $\mu_t$ is a threshold to be calibrated online. The rationale behind the probing probability \eqref{eq:p_t} is twofold: \textit{(i)} if the estimate $\hat{O}_t$ is larger than a well-designed threshold $\mu_t$, it is plausible, although not guaranteed, that the candidate action $\hat{a}_t$ would lead to an outage event; and \textit{(ii)} in order to ensure that the true outage variable $O_t$ is available at the agent sufficiently often, the agent queries the state with probability no smaller than the exploration probability $\rho$.

The threshold in \eqref{eq:p_t} is potentially updated whenever the agent decides to probe the state, i.e., all times $t$ with $P_t=1$. Based on state $s_t$ and outage indicator $O_t$ in \eqref{eq:outage_flag}, the update of the threshold $\mu_t$ follows the rule 
\begin{equation}
\label{eq:update_rule}
    \mu_{t+1}\gets \mu_t-\eta\,\tfrac{O_tP_t}{p_t} \big[(1-p_t)\,-\alpha\big],
\end{equation}
where $\eta>0$ is an update rate and $\mu_1\in[0,1]$. By \eqref{eq:update_rule}, when $P_t=0$ (no query) or $O_t=0$ (no outage), the threshold is not modified and we recover $\mu_{t+1}\gets\mu_t$.
When $P_t=1$ and $O_t=1$, i.e., in the case of a query and of an outage event, the threshold $\mu_t$ is modified as follows:

\noindent $\bullet$ If $\hat{O}_t<\mu_t$, the query is prompted by exploration, since the probing probability in \eqref{eq:p_t} is $p_t=\rho$. This implies that the estimation $\hat{O}_t$ failed to predict the outage event $O_t=1$ under the current threshold $\mu_t$. Accordingly, the threshold is decreased by the amount $\eta(1-\rho -\alpha)/\rho >0$ to facilitate the prediction of future outage events. Recall the design constraint $\rho<1-\alpha$. The intuition behind the magnitude of the update is that, since exploration already inherently identifies a fraction $\rho$ of the outages, we wish the threshold mechanism to identify the remaining fraction $1-\rho$. The scale factor $1/\rho$ plays the role of an inverse probability weighting \cite{candes2023conformalized}, compensating the fact that $O_t$ is selectively observed.

\noindent $\bullet$ If $\hat{O}_t \geq \mu_t$, the query is caused by a correct thresholding of the outage probability estimate $\hat{O}_t$. In this case, the threshold is increased by $\eta\alpha$ with the goal of reducing the probing rate. 

\noindent \textbf{Missed  Query Error Guarantees.} We now show that \gls{ocsp} guarantees condition \eqref{eq:time_b} on the \gls{mqe}. 
\begin{proposition}[Performance guarantees]
\label{prop:mqe-bound}
For any $\delta\in(0,1)$, any state sequence $\{s_t\}_{t=1}^T$, any estimation functions $\{f_t\}_{t=1}^T$ and utility functions $\{U_t\in [U_{\mathrm{min}}, U_{\mathrm{max}}]\}_{t=1}^T$, with probability at least $1-\delta$ over the algorithm's exploration mechanism in \eqref{eq:p_t}, \gls{ocsp} satisfies the upper bound
\begin{equation}
 \mathrm{MQE}_T \;\le\; \alpha+\Delta(N_T,\delta)
  \label{eq:mqe-bound}
\end{equation}
where, for $N_T\ge1$, we have defined
\[
  \Delta(N_T,\delta)=
  \frac{1 + 2\eta/\rho}{\eta N_T} + \sqrt{\frac{2(1-\rho)}{\rho N_T}\log{ \frac{2}{\delta} }} + \frac{2(1-\rho)}{3\rho N_T}\log{\frac{2}{\delta}},
\]
and $\Delta(0,\delta)=0$.
\end{proposition}

\begin{proof}
We start by rewriting \eqref{eq:mqe} as
\begin{equation}
\label{eq:two_terms}
\mathrm{MQE}_T =  \alpha + \frac{1}{N_T} \Bigg[\sum_{t=1}^T e_t + \sum_{t=1}^T O_t\frac{1-\alpha}{p_t} (p_t-P_t)\Bigg].
\end{equation}
The first sum in \eqref{eq:two_terms} is deterministically controlled by the update \eqref{eq:update_rule} in a manner similar to online conformal prediction~\cite{gibbs2024conformal}, and we omit details here due to space constraints,  while the second sum can be upper bounded with high probability with respect to the sequence of probing decisions $\{P_t\}$, concluding the proof. 


For the second sum, writing $P^{t-1}=[P_1,...,P_{t-1}]$ and $\mc{T}^{\mc{O}}=\{t:O_{t}=1\}_{t\leq T}$, we have $\mathbb{E}[P_t \mid P^{t-1}]=p_t$, and thus the sum $\sum_{t\in\mc{T}^{\mc{O}}}(p_t-P_t)/p_t$ is a martingale difference sequence. Moreover, we have $\mathrm{Var}((p_t-P_t)/p_t\mid P^{t-1}) = (1-p_t)/p_t \leq (1-\rho)/\rho$. Therefore, using Freedman's inequality \cite{freedman1975tail} on the sequence $\{(p_t-P_t)/p_t\}$, we obtain the inequality \begin{equation}
\label{eq:martingale}
    \left| \sum_{t\in\mc{T}^{\mc{O}}} \frac{p_t-P_t}{p_t} \right| \leq \sqrt{2\frac{(1-\rho)N_T}{\rho}\log\frac{2}{\delta}} + \frac{2(1-\rho)}{3\rho} \log\ \frac{2}{\delta}
\end{equation}
with probability no smaller than $1-\delta$. 


\end{proof}

\section{Experiments}
\label{sec:experiments}
\begin{figure*}
    \centering
    \subfloat{\input{figures/legend_time_plot}} \\ 
    \vspace{-0.3cm}    \setcounter{subfigure}{0}
    \addtocounter{figure}{1}
    \refstepcounter{subfigure}\label{fig:bound}\subfloat{\input{figures/all_mqe}}
    \setcounter{subfigure}{1}\refstepcounter{subfigure}\label{fig:query_rate}\subfloat{\input{figures/all_query_rate}}
    \setcounter{subfigure}{2}\refstepcounter{subfigure}\label{fig:reward}\subfloat{\input{figures/all_reward}}
    \vspace{-0.2cm}
    \setcounter{figure}{1}
    \caption{MQE, average query rate and average reward over time $t$ for $10^5$ steps. Results are averaged over 100 different seeds.}
    \vspace{-0.5cm} 
    
    \label{fig:allplots}
\end{figure*}
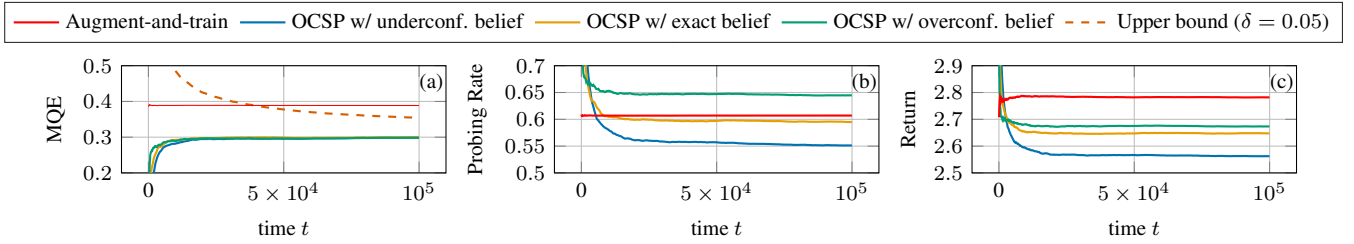

We validate \gls{ocsp} and its theoretical guarantees by investigating the \gls{mqe}, utility and query rate of \gls{ocsp} under different belief models, while comparing with the ideal benchmark in \cite{talli2025pragmatic}, which requires state augmentation and retraining. 

\noindent \textbf{Setting.} Even though our framework applies to generic, even deterministic and adversarial, dynamic systems, we consider here an \gls{mdp}, since this allows us comparisons with established benchmarks. The considered  \gls{mdp} has   $|\mc{S}|=30$ states, $|\mc{A}|=4$ actions, and a transition model $M$ obtained as follows. For each pair $(s,a)$ we randomly select a target state $s_a^*\in\mc{S}$ and define the transition probability as $M(s,a,s') \propto \exp(-d(s',s_a^*)^2/\sigma^2)$, 
where $d(s,s^*) = \min (|s-s^*|, \mc{S} - |s-s^*|)$ is a circular distance between discrete states and $\sigma^2$ = 0.67. The state sequence is drawn from this model given the control actions. The utility is derived from a reward given by the distance from a randomly selected target state $s^*$ as $r(s,a,s')=10\exp (-0.5\cdot d(s',s^*))$. Specifically, the utility $U_t(s,a)$ is obtained from a pre-trained state-action value function $Q(s,a)$ optimized by running Q-learning \cite{jaakkola1993convergence} using the implementation provided in \cite{chades2014mdptoolbox} over the given \gls{mdp}.

\noindent \textbf{Implementation.} When $P_{t-1}=0$, the belief $b_t$ is obtained from an estimated transition model $\tilde M$ by updating the last belief $b_{t-1}$ given action $a_{t-1}$ as  $f_t(b_{t-1},a_{t-1}) = b_{t-1}^\intercal \tilde M_{a_{t-1}},$ 
where $\tilde M_{a_{t-1}}$ is the $|\mc{S}|\times|\mc{S}|$ transition matrix associated with action $a_{t-1}$. The estimated transition model $\tilde M$ is defined using the same moder as for $M$, but with a generally different dispersion parameter $\sigma^2$. Specifically we considered three settings corresponding to predictive models with different calibration levels \cite{guo2017calibration, simeone2022machine, huang2025distilling}: \textit{(i)} \textit{exact belief}, i.e., $\sigma^2=0.67$; \textit{(ii)} \textit{overconfident belief} with a lower spread of $\sigma^2=0.37$; and \textit{(iii)} \textit{underconfident belief} with the larger spread $\sigma^2=2$. 

\noindent \textbf{Baseline.} As a baseline, we consider the ``augment-and-train'' methodology in \cite{talli2025pragmatic}, which trains from scratch in the \gls{mdp} a policy that simultaneously learns how to act and probe. This policy operates on an augmented state space, describing also the time elapsed from the last probed state, and it is trained to meet an average probing rate. Augment-and-train serves as an upper bound on the performance of \gls{ocsp} for a given probing rate, since \gls{ocsp} operates as a wrapper around a pre-trained value function.

In Fig.~\ref{fig:allplots} we plot $\mathrm{MQE}_T$, the average probing rate and the average normalized cumulative return, as a function of time $t$. Fig.~\ref{fig:bound} demonstrates that \gls{ocsp} can control the \gls{mqe} to meet the target rate $\alpha=0.3$, irrespective of the quality of the prediction without requiring training. In contrast, the augment-and-train scheme, which is trained for the given target probing rate, which is set to approximately match that of \gls{ocsp} with an exact belief (see Fig.~\ref{fig:query_rate}), yields an \gls{mqe} as high as 0.4. The figure also reports the bound from Proposition~\ref{prop:mqe-bound} with $\delta=0.05$ to further validate the theoretical guarantees of \gls{ocsp}. 

In Fig.~\ref{fig:query_rate} we observe that \gls{ocsp} tends to concentrate state probings in the first time steps, while stabilizing to a lower probing rate as time goes on. Specifically, an underconfident belief induces candidate actions (\ref{eq:candidate_a}) that cater to a larger set of possible states, thus requiring \gls{ocsp} to issue fewer probing actions to compensate for regret events. In contrast, with an overconfident belief, the candidate actions  (\ref{eq:candidate_a})  may make  riskier choices tailored to a particular state, and \gls{ocsp} compensates for resulting outage events  by increasing the probing rate. 

The different probing rates translate into different cumulative returns, as shown in Fig.~\ref{fig:reward}. In particular, Fig.~\ref{fig:reward} shows that \gls{ocsp} with an overconfident belief obtains a slightly higher average return than \gls{ocsp} with an exact belief but with higher average probing rate. 

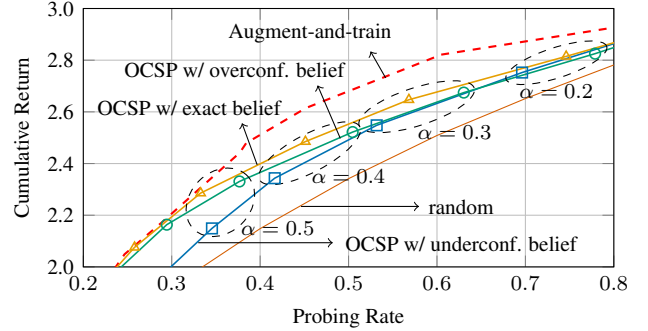
\begin{figure}
    \centering
    \input{figures/pareto_plot/pareto_plot}
    \caption{Cumulative reward vs. probing rate for different values of $\alpha$.}
    \vspace{-0.5cm}
    \label{fig:pareto}
\end{figure}

The trade-off between probing rate and cumulative return is further investigated in Fig.~\ref{fig:pareto}, which shows the Pareto front on the plane with axes given by cumulative return and probing rate, both evaluated at time $t=10^5$. The Pareto front is obtained for \gls{ocsp} by varying the \gls{mqe} target $\alpha$, while for augment-and-train we varied the target probing rate. \gls{ocsp}, at lower query rates, is seen to almost match the performance of the ideal augment-and-train scheme. Using an exact belief is observed to be instrumental for achieving a better performance trade-off with \gls{ocsp}. Moreover, an overconfident belief model is preferable to an underconfident one at low query rates, while, when the query rate is higher than 0.7, the underconfident belief obtains a higher trade-off. 

\section{Conclusions}
In this work, we developed a sequential decision-making framework with adaptive state probing and formal performance guarantees. Future work may focus on  larger-scale evaluations and on applications to AI-based policy models.

\newpage

\bibliographystyle{IEEEbib}
\bibliography{mybiblio}

\end{document}

%% file: acronyms.tex
\newacronym{mdp}{MDP}{Markov decision process}
\newacronym{rl}{RL}{Reinforcement Learning}
\newacronym{mqe}{MQE}{missed query error}
\newacronym{mpi}{MPI}{Modified Policy Iteration}
\newacronym{ocsp}{OCSP}{online conformal state probing}
\newacronym{ai}{AI}{Artificial Intelligence}
\newacronym{llm}{LLM}{large language model}
\newacronym{vla}{VLA}{vision language action model}

%% file: figures/legend_time_plot.tex
\begin{tikzpicture}

\definecolor{darkgray176}{RGB}{176,176,176}

\begin{axis}[
width=0cm,
height=0.1cm,
scale only axis,
tick align=inside,
xmin=0, xmax=0,
ymin=0, ymax=0.1,
legend style={
    draw=white!15!black,
    at={(0, 0)}, anchor=south,
    /tikz/every even column/.append style={column sep=0.2em}
},
legend columns=6
]

\addplot[red, thick] table {
0    0.1
};

\addplot[okabeBlue, , thick] table {
0    0.1
};

\addplot[okabeOrange, thick] table {
0    0.1
};

\addplot[okabeGreen, thick] table {
0    0.1
};

\addplot[okabeVermillion, thick, dashed] table {
0    0.1
};

\legend{\footnotesize Augment-and-train, \footnotesize OCSP w/ underconf. belief, \footnotesize OCSP w/ exact belief, \footnotesize OCSP w/ overconf. belief, \footnotesize Upper bound ($\delta=0.05$)}
\end{axis}
\end{tikzpicture}

%% file: figures/all_mqe.tex
\begin{tikzpicture}
            \begin{axis}[
                width=0.33\textwidth,
                height=3cm,
                grid=major,
                xlabel={time $t$}, 
                ylabel={MQE},
                ymin=0.2, ymax=0.5,
                xtick={0,50000,100000},
                scaled ticks=false,
                xticklabels = {0,$5\times10^4$, $10^5$},
                label style={font=\footnotesize},       
                tick label style={font=\footnotesize}, 
            ]
                \addplot[okabeBlue, thick] table [x=time_step, y=mean_mqe, col sep=comma] {figures/time_plot/0.5.csv};
                \addplot[okabeOrange, thick] table [x=time_step, y=mean_mqe, col sep=comma] {figures/time_plot/1.5.csv};
                \addplot[okabeGreen, thick] table [x=time_step, y=mean_mqe, col sep=comma] {figures/time_plot/2.7.csv};
                
                \addplot[okabeVermillion, thick, dashed] table [x=time_step, y=bound, col sep=comma]{figures/time_plot/0.5.csv};

                \addplot[red] table [x=time_step, y=mqe_comparison, col sep=comma]{figures/time_plot/0.5.csv};
            \end{axis}

            \node[  
                   ] at (4.1,1.2) {\footnotesize (a)};
        \end{tikzpicture}

%% file: figures/all_query_rate.tex
\begin{tikzpicture}
            \begin{axis}[
                width=0.33\textwidth,
                height=3cm,
                grid=major,
                xlabel={time $t$},
                ylabel={Probing Rate},
                ymin=0.5, ymax=0.7,
                xtick={0,50000,100000},
                scaled ticks=false,
                xticklabels = {0,$5\times10^4$, $10^5$},
                label style={font=\footnotesize},       
                tick label style={font=\footnotesize}, 
            ]
                \addplot[okabeBlue, thick] table [x=time_step, y=mean_query_rate, col sep=comma] {figures/time_plot/0.5.csv};
                \addplot[okabeOrange, thick] table [x=time_step, y=mean_query_rate, col sep=comma] {figures/time_plot/1.5.csv};
                \addplot[okabeGreen, thick] table [x=time_step, y=mean_query_rate, col sep=comma] {figures/time_plot/2.7.csv};
                
                \addplot[red, thick] table [x=time_step, y=query_rate_comparison, col sep=comma]{figures/time_plot/0.5.csv};
            \end{axis}

            \node[  
                   ] at (4.1,1.2) {\footnotesize (b)};
        \end{tikzpicture}

%% file: figures/all_reward.tex
\begin{tikzpicture}
            \begin{axis}[
                width=0.33\textwidth,
                height=3cm,
                grid=major,
                xlabel={time $t$},
                ylabel={Return},
                ymin=2.5, ymax=2.9,
                xtick={0,50000,100000},
                scaled ticks=false,
                xticklabels = {0,$5\times10^4$, $10^5$},
                label style={font=\footnotesize},       
                tick label style={font=\footnotesize}, 
            ]
                \addplot[okabeBlue, thick] table [x=time_step, y=mean_rewards, col sep=comma] {figures/time_plot/0.5.csv};
                \addplot[okabeOrange, thick] table [x=time_step, y=mean_rewards, col sep=comma] {figures/time_plot/1.5.csv};
                \addplot[okabeGreen, thick] table [x=time_step, y=mean_rewards, col sep=comma] {figures/time_plot/2.7.csv};
                
                \addplot[red, thick] table [x=time_step, y=reward_comparison, col sep=comma]{figures/time_plot/0.5.csv};
            \end{axis}

            \node[  
                   ] at (4.1,1.2) {\footnotesize (c)};
        \end{tikzpicture}

%% file: figures/pareto_plot/pareto_plot.tex
\begin{tikzpicture}
            \begin{axis}[
                width=\linewidth,
                height=5cm,
                grid=major,
                xlabel={Probing Rate}, 
                ylabel={Cumulative Return},
                legend pos=north west,
                xmin=0.2, xmax=0.8,
                ymin=200.0, ymax=300.0,
                ytick={200,220,240,260,280,300},
                yticklabels={2.0,2.2,2.4,2.6,2.8,3.0},
                label style={font=\footnotesize},       
                tick label style={font=\footnotesize}, 
            ]

                \addplot[red, thick, dashed] table [x=query_rate, y=reward, col sep=comma] {figures/pareto_plot/noisy_plot/curve_noise1.5_optimal.csv};
            
                \addplot[okabeBlue, semithick, mark=square] table [x=query_rate, y=reward, col sep=comma] {figures/pareto_plot/noisy_plot/curve_noise0.5_alpha0.0_beta0.0.csv};
                \addplot[okabeOrange, semithick, mark=triangle] table [x=query_rate, y=reward, col sep=comma] {figures/pareto_plot/noisy_plot/curve_noise1.5_alpha0.0_beta0.0.csv};
                \addplot[okabeGreen, semithick, mark=o] table [x=query_rate, y=reward, col sep=comma] {figures/pareto_plot/noisy_plot/curve_noise2.7_alpha0.0_beta0.0.csv};

                \addplot[okabeVermillion] table {
                0.3 192.108
                0.4 214.819
                0.5 234.149
                0.6 250.86
                0.7 265.14
                0.8 278.150
                };

                \draw[rotate around={18:(0.735,279)}, dashed] (0.735,279) ellipse (0.7cm and 0.2cm);
                \node at (0.735,268) {\footnotesize $\alpha=0.2$};

                \draw[rotate around={18:(0.577,262)}, dashed] (0.577,262) ellipse (0.8cm and 0.25cm);
                \node at (0.62,252) {\footnotesize $\alpha=0.3$};

                \draw[rotate around={28:(0.456,244)}, dashed] (0.456,244) ellipse (0.75cm and 0.25cm);
                \node at (0.5,235) {\footnotesize $\alpha=0.4$};

                \draw[rotate around={45:(0.355,225)}, dashed] (0.355,225) ellipse (0.5cm and 0.4cm);
                \node at (0.42,215) {\footnotesize $\alpha=0.5$};

             \end{axis}

             \node[
                   ] at (3,3.1) {\footnotesize Augment-and-train};
            \node[
                   ] at (1.4,2.1) {\footnotesize OCSP w/ exact belief};
            \node[
                   ] at (2,2.6) {\footnotesize OCSP w/ overconf. belief};
            \node[
                   ] at (5,0.3) {\footnotesize OCSP w/ underconf. belief};
            \node[
                   ] at (5,0.8) {\footnotesize random};

            \draw[->] (4,2.50) -- (3.8,2.85);
            \draw[->] (3.4,1.7) -- (3.1,2.4);
            \draw[->] (2.3,1.35) -- (2.1,1.9);
            \draw[->] (1.52,0.32) -- (3.3,0.32);
            \draw[->] (2.88,0.8) -- (4.44,0.8);
             
        \end{tikzpicture}